%% file: Scalable_Security_Allocation.tex
\documentclass[9pt,shortpaper,twoside,web]{ieeecolor}
\usepackage{generic}
\usepackage{cite}
\usepackage[hidelinks]{hyperref}
\usepackage{amsmath,amssymb,amsfonts}
\usepackage{algpseudocode,algorithm}
\usepackage{graphicx}
\usepackage{dsfont}
\usepackage{textcomp}
\usepackage{subcaption}
\usepackage{mathtools}
\usepackage{xcolor}
\usepackage{tabularx}
\usepackage{academicons}
\usepackage{multirow}
\usepackage{diagbox}
\usepackage[colorinlistoftodos]{todonotes}
\input{deflatex3}
\newtheorem{Lemma}{Lemma}
\newtheorem{Theorem}{Theorem}
\newtheorem{Definition}{Definition}
\newtheorem{Assumption}{Assumption}
\newtheorem{Remark}{Remark}
\newtheorem{Problem}{Problem}
\newtheorem{Proposition}{Proposition}

\newtheorem{Corollary}{Corollary}

\makeatletter
\newcommand{\multiline}[1]{%
  \begin{tabularx}{\linewidth-\ALG@thistlm-0.0cm}[t]{@{}X@{}}
    #1
  \end{tabularx}
}
\makeatother
\def\BibTeX{{\rm B\kern-.05em{\sc i\kern-.025em b}\kern-.08em
T\kern-.1667em\lower.7ex\hbox{E}\kern-.125emX}}
\begin{document}
\title{Scalable and Optimal Security Allocation in Networks against \\ Stealthy Injection Attacks
}
\author{Anh Tung Nguyen, Sribalaji C. Anand, André M. H. Teixeira 
\thanks{This work is supported by the Swedish Research Council under the
grants 2021-06316, 2024-00185  
and by the Swedish Foundation for Strategic Research.}
\thanks{Anh Tung Nguyen and Andr{\'e} M. H. Teixeira are with the Department of Information Technology, Uppsala University, PO Box 337, SE-75105, Uppsala, Sweden (e-mail: \{anh.tung.nguyen, andre.teixeira\}@it.uu.se).}
\thanks{ Sribalaji C. Anand is with the School of Electrical Engineering and Computer Science, KTH Royal Institute of Technology, 1827 Stockholm, Sweden (e-mail: srca@kth.se). } 
}
	
\maketitle
	
\begin{abstract}
This paper addresses the security allocation problem in a networked control system under stealthy injection attacks. The networked system is comprised of interconnected subsystems which are represented by nodes in a digraph. An adversary compromises the system by injecting false data into several nodes with the aim of maximally disrupting the performance of the network while remaining stealthy to a defender. 
To minimize the impact of such stealthy attacks, the defender, with limited knowledge about attack policies and attack resources, allocates several sensors on nodes to impose the stealthiness constraint governing the attack policy.
We provide an optimal security allocation algorithm to minimize the expected attack impact on the entire network.
Furthermore, under a suitable local control design, the proposed security allocation algorithm can be executed in a scalable way.
Finally, the obtained results are validated through several numerical examples.
\end{abstract}
	
\section{Introduction}
Networked control systems are prevalent in modern society and found in critical infrastructure such as power grids, transportation systems, and water distribution networks \cite{molzahn2017survey,polycarpou2023smart,conti2021survey}. 
These systems commonly operate 
through open 
communication technologies such as the public internet and wireless communications, making them susceptible to cyber attacks \cite{teixeira2015secure,falliere2011w32,kshetri2017hacking}. Such attacks can have dire financial and societal repercussions. Notable examples include the Iranian industrial control system and the Ukrainian power grid, which suffered devastating impacts from  Stuxnet in 2010 \cite{falliere2011w32} and Industroyer in 2016 \cite{kshetri2017hacking}, respectively. Consequently, the importance of security in control systems has reached an unprecedented level in response to these significant threats.

In an attempt to handle cyber attacks, various studies find defense mechanisms by solving optimization problems that formulate the objectives of the malicious adversary and the defender simultaneously. By casting the security allocation problem in the Stackelberg game framework, the authors in \cite{li2018false} aim to find the game equilibrium by formulating a single optimization problem that accounts for all admissible adversarial actions with a simplified linear mapping from action spaces to game payoffs.
Similarly, Yuan et al. also seek the Stackelberg equilibrium by considering all admissible adversarial actions in the context of industrial control systems in \cite{yuan2019stackelberg}. 
It is worth noting that such a linear mapping may not cover advanced attack policies such as stealthy attacks \cite{umsonst2021bayesian}, which will be addressed in this paper.
The work in \cite{shukla2022robust} adopts a slightly different game setup, where both the defender and the adversary have discrete action spaces. The setup allows for the use of the traditional backward induction to determine the Stackelberg equilibrium. However, this requires evaluating all admissible players' actions, leading to very high computational complexity, as noted in \cite[Section V]{shukla2022robust}. Instead of using the Stackelberg game framework, 
the authors in \cite{umsonst2021bayesian,pirani2021game} allow the malicious adversary and the defender make their decisions simultaneously. It is worth noting that a pure equilibrium in such a scenario does not necessarily exist. Alternatively, a mixed-strategy equilibrium notation was employed to outline the defense strategy, 
which is also known as moving target defense. As presented in \cite{umsonst2021bayesian}, the mixed-strategy equilibrium requires an explicit game matrix that accommodates all the possible pairs of players' actions, probably resulting in a high computational complexity.

One interesting observation among the existing studies \cite{li2018false,yuan2019stackelberg,shukla2022robust,umsonst2021bayesian,pirani2021game} is that no analysis and computation of the performance loss caused by stealthy attacks (also henceforth referred to as attack impact) on the network has been made. 
This issue has been addressed in our previous study \cite{nguyen2024security} by adopting the output-to-output gain security metric \cite{teixeira2015strategic}. Our previous work \cite{nguyen2024security} mainly studies the graph-theoretical condition under which the security metric is well-posed while this paper studies how to efficiently exploit the security metric given its well-posedness. Moreover, this paper extends a single attack node in \cite{nguyen2024security} to multiple attack nodes.
Another observation is that explicitly computing the game payoff for all the admissible players' actions results in an enormous computational complexity. 
These two issues are studied in this paper.

In this paper, we consider a continuous-time networked control system, associated with a strongly connected digraph, under stealthy attacks. 
The system consists of several interconnected one-dimensional subsystems, known as nodes in the digraph. The adversary aims to maximally degrade the global network performance 
by selecting several nodes to execute stealthy data injection attacks on their inputs. 
Concurrently, the defender chooses a set of nodes and monitors their outputs with the goal of detecting and mitigating the impact of attacks. The goal of this paper is to find the optimal set of nodes for the defender in an efficient way. We study this security allocation problem and provide the following contributions: 
\begin{itemize}
    \item The worst-case impact of stealthy attacks on the global network performance is evaluated via a novel Attack-Energy-Constrained Output-to-Output gain security metric.
    \item The well-posedness of the proposed security metric is guaranteed and the metric is shown to be equivalent to a semidefinite program (SDP) problem.
    \item Without explicitly computing a large matrix game as in  \cite{li2018false,yuan2019stackelberg,shukla2022robust,umsonst2021bayesian}, we propose a novel optimization-based method to find the optimal monitor nodes that minimize the security metric.
    \item Further, we show that the optimal monitor nodes can be exactly found in a more scalable optimization problem under a suitable design of local controllers.
\end{itemize}

The remainder of the paper is the following. Section~\ref{sec:problem_formulation} formulates the security allocation problem after presenting the description of the network and the models of the adversary and the defender. We analyze the security allocation problem and propose a method that computes the optimal monitor nodes in Section~\ref{sec:general_solution}. In Section~\ref{sec:impatient_adv}, a scalable method is proposed to find the optimal monitor nodes in a more efficient way, suiting large networks. Section~\ref{sec:simulation} presents numerical examples to validate the obtained results while Section~\ref{sec:conclusion} concludes the paper. We conclude this section by providing the notation to be used throughout this paper.

\textbf{Notation:} $\Rbb^n \, (\Rbb^n_{>0},\,\Rbb^n_{\geq 0})$ and $\Rbb^{n \times m}$ stand for sets of real (positive, non-negative) $n$-dimensional vectors and real $n$-by-$m$ matrices, respectively; 
the set of $n$-by-$n$ symmetric positive semidefinite matrices is denoted as $\Sbb^n_{\succeq 0}$.
We denote $A \succ(\succeq) B$ if $A-B$ is a positive definite (semi-definite) matrix.
An $i$-th column of the $n$-by-$n$ identity matrix is denoted as $e_i$.
For a set $\Ac$, $|\Ac|$ stands for the set cardinality and a collection of variables $(x_{a_1},x_{a_2},\ldots,x_{a_{|\Ac|}})$ where $a_i \in \Ac$ is denoted as $\{x_{a_i}\}_{\forall a_i \in \Ac}$. 
The space of square-integrable  functions is defined as $\Lc_{2} \triangleq \bigl\{f: \Rbb_{>0} \rightarrow \Rbb^n ~|~ \norm{f}^2_{\Lc_2 [0,\infty]} < \infty \bigr\}$ and the extended space be defined as $\Lc_{2e} \triangleq \bigl\{ f: \Rbb_{>0} \rightarrow \Rbb^n ~|~ \norm{f}^2_{\Lc_2 [0,H]} < \infty,~ \forall~ 0 < H < \infty \bigr\} $ where $\norm{f}_{\Lc_2 [0,H]}^2 \triangleq \int_{0}^{H} \norm{f(t)}_2^2 \, \text{d}t$.
The notation $\norm{f}^2_{\Lc_2}$  is used  as shorthand for the norm $\norm{f}_{\Lc_2 [0,H]}^2$ if the time horizon $[0,H]$ is clear from the context.
Let $\Gc \triangleq (\Vc, \Ec, A, \Theta)$ be a digraph with the set of $N$ nodes $\Vc = \{1, 2,...,N\}$, the set of edges $\Ec \subseteq \Vc \times \Vc $, the adjacency matrix $A = [A_{ij}]$, and self-loop gain diagonal matrix $\Theta = \textbf{diag}([\theta_i])$ where \textbf{diag} stands for a diagonal matrix.
Each element $(i,j) \in \Ec, ~i\neq j$, represents a directed edge from $j$ to $i$ and 
the element $A_{ij}$ of the adjacency matrix is positive, and with $(i,j) \notin \Ec$ or $i = j$, $A_{ij} = 0$. 
The in-degree Laplacian matrix is defined as $L = [\ell_{ij}] = \Theta + \textbf{diag}(A \textbf{1}) - A$ where $\textbf{1}$ is an all-one vector.
Further, $\Gc$ is called a strongly connected digraph if and only if $\sum_{k = 1}^{N-1} A^k$ has no zero entry.
The set of all in-neighbours of node $i$ is denoted as $\Nc_i = \{j \in \Vc~|~ (i,j) \in \Ec \}$.

\section{Problem Formulation}
\label{sec:problem_formulation}
In this section, we first describe a networked control system, with a global performance metric, under stealthy attacks. 
Secondly, the adversary and the defender are modeled with opposing objectives.
Finally, we formulate the two main research problems that will be studied in this paper. 
\subsection{Networked control systems under attacks}
In this subsection, we introduce a networked control system in a normal operation without attacks and describe the resources of the adversary and the defender. Subsequently, the system modeling under attacks is formulated.  
\subsubsection{Networked control systems without attacks}
Consider a strongly connected digraph $\Gc \triangleq (\Vc, \, \Ec, \, A, \, \Theta)$ with $N$ nodes, the state-space model of a one-dimensional node $i$ is described:
\begin{align}
	\dot x_i(t) &= A_{ii} x_i(t) + u_i(t), 
 ~ i \in \Vc = \bigl\{1,\,2,\ldots,\,N\bigr\},
	\label{sys:xi}
\end{align}
where $x_i(t) \in \Rbb$ and $u_i(t) \in \Rbb$ are the state and the local control input, respectively. The local parameter $A_{ii}$ is given to the operator. 
Each node $i \in \Vc$ is controlled by the following local control law:
\begin{align}
	u_i(t) = -(\theta_i + A_{ii}) x_i(t) +  \sum_{j \, \in \, \Nc_i} A_{ij} \big(x_j(t) - x_i(t)\big),
	\label{sys:ui}
\end{align}
where $\theta_i \in \Rbb_{>0}$ 
is the gain of a self-loop at node $i$. 
For convenience, let us denote $x(t)$ as the state of the entire network, $x(t) \triangleq \big[x_1(t),~x_2(t),\ldots,~x_N(t)\big]^\top$. Therefore, the closed-loop network control system \eqref{sys:xi}-\eqref{sys:ui} without attacks can be rewritten as follows:
\begin{align}
    \dot x(t) = -L x(t), \label{sys:x_healthy}
\end{align}
where $L$ is the in-degree Laplacian matrix representing the graph $\Gc$.


Similar to robust control, the global performance of the entire network for a given, possibly infinite, time horizon $[0,H]$ is formulated as follows:
\begin{align}
    J \triangleq \norm{p}_{\Lc_2[0,H]}^2.
    \label{sys:J}
\end{align}
Here, the performance output vector $p(t) \in \Rbb^N$ is defined as follows:
\begin{align}
    p(t) = W x(t), \label{sys:pi}
\end{align}
where $W = \textbf{diag}([w_i])$ and $w_i \in \Rbb_{> 0}$ is a given weighting factor. The weighting factor is determined by the operator. It is worth noting that this weighting factor will be discussed more in Section~\ref{sec:impatient_adv}.


\subsubsection{Defense resources}
To get prepared for facing malicious activities, the defender selects a subset of the node set $\Vc$ as a set of monitor nodes, denoted as $\Mc = \{m_1,m_2,\ldots,m_{| \Mc |} \} \subset \Vc$. A sensor is placed on each monitor node to monitor its output, where the number of utilized sensors should be constrained for practical reasons.
Let us denote $\beta \in \Rbb_{>0}$ as the sensor budget that is the maximum number of utilized sensors, i.e., 
\begin{align}
    |\Mc| \leq \beta. \label{sensor_budget}
\end{align} 
More specifically, the defender monitors the following output measurements:
\begin{align}
    y_m(t) = e^\top_m x(t),~~ \forall \, m \in \Mc. \label{sys:ym}
\end{align}
At each monitor node $m \in \Mc$, a corresponding alarm threshold $\delta_m \in \Rbb_{>0}$ is assigned. The defender notifies the presence of the adversary if the output energy for a given time horizon $[0,H]$ of at least one monitor node crosses its corresponding alarm threshold, i.e.,
\begin{align}
    \norm{y_m}^2_{\Lc_2[0,H]} > \delta_m. \label{alarm_deltam}
\end{align}
For later uses, let us denote the vector of all the alarm thresholds as $\delta \triangleq [\delta_1,\,\delta_2,\ldots,\,\delta_N]^\top$.
Further, each node $i$ has a cost $\kappa_i \in \Rbb_{>0}$, resulting in the following sensor cost of the monitor set $\Mc$:
\begin{align}
    c_s(\Mc) = \sum_{m \, \in \, \Mc} \kappa^\top e_m, \label{cost_sensor}
\end{align}
where $\kappa = \left[ \kappa_1, \kappa_2, \ldots, \kappa_N \right]$ is the given vector of all sensor costs.
\subsubsection{Adversary resources}
In critical conditions, the adversary gains access to the network and acquires enough information about the system parameters and defense strategies. Then, the adversary conducts the following attack strategies.

The adversary selects exactly $\alpha~(\alpha \leq N)$ attack nodes on which to conduct stealthy data injection attacks on their inputs. 
Let us denote the set of attack nodes as $\Ac \, (|\Ac| = \alpha)$.
For each attack node $a \in \Ac$, the adversary constructs an additive attack signal $\zeta_a(t)$, which is assumed to have bounded energy: 
\begin{align}
   \norm{\zeta_a}^2_{\Lc_2[0,\infty]}  \leq E < \infty, ~\forall \, a \in \Ac, \label{zeta_bounded}
\end{align}
where the maximum attack energy $E$ is given. 
Then, the control input in \eqref{sys:ui} under false data injection attacks can be represented as follows:
\begin{align}
    \tilde u_a(t) = u_a(t) + \begin{cases}
        \zeta_a(t),~~ \text{if}~~ a \in  \Ac, \\
        0, ~~ \text{otherwise}. 
    \end{cases} \label{sys:uia}
\end{align}

Given a set of attack nodes $\Ac \triangleq \{a_1, a_2, \ldots, a_{\alpha} \}$, let us denote the corresponding attack input matrix as $B_\Ac \triangleq [e_{a_1},e_{a_2},\ldots,e_{a_{\alpha}}]$ and the attack signal vector as $\zeta(t) \triangleq [\zeta_{a_1}(t),\zeta_{a_2}(t),\ldots,\zeta_{a_{\alpha}}(t)]^\top$.

\subsubsection{Networked control systems under attacks}
Given the above descriptions of the healthy network \eqref{sys:x_healthy}, the healthy global performance \eqref{sys:J}-\eqref{sys:pi}, the monitor outputs \eqref{sys:ym}, the adversary strategy \eqref{sys:uia}, and a set of attack nodes $\Ac$ which corresponds to an attack input matrix $B_\Ac$, the network control system under the stealthy attack can be described as follows:
\begin{align}
    \dot x^a(t) &= - L x^a(t)  + B_\Ac  \zeta(t), \label{sys:xa} \\
    p^a(t) &= W x^a(t), \label{sys:pa} \\
    y^a_m(t) &= e_m^\top x^a(t), ~ \forall m \in \Mc, \label{sys:yma}
\end{align}
where $W \triangleq \textbf{diag}([w_i])$ and the superscript ``a'' stands for signals subjected to attacks. 
As a consequence, the impact of attacks on the global performance is formulated as follows:
\begin{align}
    J^a \triangleq \norm{p^a}_{\Lc_2[0,H]}^2.
    \label{sys:Ja}
\end{align}
It is worth noting that the global performance under attacks \eqref{sys:Ja} is computed based on the signals under attacks \eqref{sys:pa} and is different from the global performance \eqref{sys:J} without attacks.

Given that the digraph $\Gc$, representing the network, is strongly connected and self-loop control gains in \eqref{sys:ui} exist, the in-degree Laplacian matrix $L$ in \eqref{sys:xa} is Hurwitz. As a result, the system can be assumed to converge to its equilibrium before being exposed to attacks.
Let us make use of the following assumption.
\begin{Assumption}
    \label{assumption:x0}
    The system \eqref{sys:xa} is at its equilibrium $x_e = 0$ before being affected by the attack signal $\zeta(t)$.  \QET
\end{Assumption}

It is worth noting that the bounded energy assumption \eqref{zeta_bounded} implies that the attack signals can last for a long but finite time, say $[0,T_a],~T_a < \infty$. Although the time horizon of attack signals is unknown \textit{a priori}, the bounded energy assumption \eqref{zeta_bounded} enables us to consider that $\zeta_a(\infty) \triangleq \lim_{t \ra \infty} \zeta_a(t) = 0 ~ \forall \, a \in \Ac$. This analysis results in the network state at infinity $x^a(\infty) \triangleq \lim_{t \ra \infty} x^a(t)$ in the following lemma.
\begin{Lemma}[Limited end-point attacks]
    \label{lem:finite:horizon}
    Consider the system \eqref{sys:xa},
    it holds $x^a(\infty) = 0$ if, and only if, $\zeta_a(\infty) = 0,~\forall a \in \Ac$. \QET
\end{Lemma}
\begin{proof}
    (If) based on the discussion right before Lemma~\ref{lem:finite:horizon}, one considers that $\zeta(t) = 0, \forall t \in (T_a,\,\infty)$. Since $L$ in \eqref{sys:xa} is Hurwitz, the system \eqref{sys:xa} is asymptotically stable in the absence of an input signal, resulting in $x^a(\infty) = 0$. (Only if) the condition $x^a(\infty) = 0$ implies that $\lim_{t \ra \infty} x^a(t) = 0$ and $\lim_{t \ra \infty} \dot x^a(t) = 0$, leading to $\lim_{t \ra \infty} B_\Ac \zeta(t) = 0$ based on the dynamics \eqref{sys:xa}. Since matrix $B_\Ac$ is full column rank by construction, one has $\zeta_a(\infty) = 0$.
\end{proof}

Assumption 1 and Lemma~\ref{lem:finite:horizon} allow us to analyze the impact of the stealthy attacks subjected to the two constraints $x^a(0) = 0$ and $x^a(\infty) = 0$ throughout the remainder of the paper.
In the following, we conduct a deeper investigation into the impact of stealthy attacks, resulting in a novel security metric.
\subsection{Adversary model}
The purpose of the adversary is to maximally disrupt the global performance of the network \eqref{sys:Ja} subject to the network model \eqref{sys:xa}-\eqref{sys:yma} while remaining stealthy to the defender {(see the discussion on the importance of the stealthiness in \cite[Sec. II.E]{umsonst2021bayesian}).
The above adversarial purpose allows us to mainly focus on stealthy attacks which will be defined in the following. 
\begin{Definition}
    [Stealthy injection attacks] \label{def:stealthy_attacks}
    Consider the structure of the continuous LTI system
    \eqref{sys:xa}-\eqref{sys:yma} with monitor outputs $y_m^a(t) = e_m^\top x^a(t)$ for every $m \in \Mc$, which is a set of monitor nodes. The attack $\zeta(t)$ on the system \eqref{sys:xa}-\eqref{sys:yma} is defined as a stealthy attack if the following condition $\norm{y_m^a}^2_{\Lc_2} \leq \delta_m$ holds for all $m \in \Mc$. \QET
\end{Definition}

The stealthy attack policy in Definition~\ref{def:stealthy_attacks} enables the adversary to evaluate the maximum disruption on the network given a set of attack nodes $\Ac$ and a set of monitor nodes $\Mc$. Inspired by game theory \cite{bacsar1998dynamic}, the set of attack nodes $\Ac$ can be varied to gain the maximum disruption on the network. As a result, given a set of monitor nodes $\Mc$ and the number of attack nodes $\alpha$, the adversary solves the following optimization problem to find the optimal attack set $\Ac$ that yields the highest maximum disruption on the network:
\begin{align}
    Q(\Mc \,|\,\alpha) \triangleq \max_{\Ac \, \subset \, \Vc, \, |\Ac| =  \alpha} &~ V(\Mc, \Ac),  \label{Q_max} 
\end{align}
where $V(\Mc, \Ac )$ is the worst-case disruption on the global performance of the network \eqref{sys:Ja} and is formulated as follows:
\begin{align}
    V(\Mc, \Ac) \triangleq 
    ~\sup_{\zeta}&~
    \norm{p^a}_{\Lc_2}^2 
    \label{Q_sup} \\
    \text{s.t.}&~
    \eqref{sys:xa}-\eqref{sys:yma}, \, x^a(0) = 0, \, x^a(\infty) = 0, \non \\
    &\norm{y_m^a}^2_{\Lc_2} \leq \delta_m,~ \forall \, m \in \Mc, \non \\
    &\norm{e_j^\top \zeta}^2_{\Lc_2} \leq E,~ \forall \, j  \in \{1,2,\ldots, |\Ac| \}. \non    
\end{align}

The worst-case disruption \eqref{Q_sup} is an extended version of the Output-to-Output gain security metric proposed in \cite{teixeira2015strategic}. This security metric \eqref{Q_sup} is 
called an Attack-Energy-Constrained Output-to-Output gain security metric for a given set of attack nodes $\Ac$ and a given set of monitor nodes $\Mc$, which plays a crucial role in evaluating the impact of stealthy attacks and finding the optimal monitor set for the defender in the remainder of this paper.
Next, we model the defender who does not have full knowledge of the adversary.
\begin{Remark}
    In \eqref{Q_sup}, the first constraint is the system modeling of the network while the second and the third constraints are based on the assumption of the attack resources (Assumption~\ref{assumption:x0} and Lemma~\ref{lem:finite:horizon}). The last two constraints are the stealthiness condition \eqref{alarm_deltam} and the bounded energy \eqref{zeta_bounded} of the attack policy, respectively.
        \QET
\end{Remark}
\subsection{Defender model}
\label{sec:defender_model}
In practice, the defender seldom foresees the number of attack nodes, which is refereed to as attack type,
they need to deal with. In line with \cite{umsonst2021bayesian}, the defender can know several possible attack types with some probabilities. Let us assume that the defender considers a set of $n_a~ (n_a \leq N)$ attack types, which is defined as follows:
\begin{align}
    \Tc \triangleq \{ \alpha_1, \, \alpha_2,\ldots,\, \alpha_{n_a} \}, \label{attack_type}
\end{align} 
where an attack type $\alpha_k~(0 < \alpha_k \leq N)$ has exactly $\alpha_k$ attack nodes. Here, the attack types in $\Tc$ are modeled as random variables, where the probability of each type $\alpha_k$ is known \textit{a priori} and denoted as $\phi_k~ (0 \leq \phi_k \leq 1)$ and $\sum_{k = 1}^{n_a}~\phi_k = 1$. 
The information about these probabilities can be provided by conducting a risk assessment \cite{9011}, which is in line with the Bayesian and the stochastic game settings \cite{umsonst2021bayesian,sayin2021bayesian,etesami2018stochastic}. 
For each attack type $\alpha_k$, let us denote $\Sc_{\alpha_k}$ as a collection of all the admissible attack sets, i.e., 
\begin{align}
    \Sc_{\alpha_k} \triangleq \{ \Ac \subseteq \Vc \, | \, |\Ac| = \alpha_k \}. \label{col_attack_set}
\end{align}

In risk management, the fewer attack resources adversaries need to mount an attack, the more likely such attacks are to occur, leading to the following reasonable assumption.
\begin{Assumption}
    \label{assumption:adversary_type}
    For any pair of attack types $\alpha_k$ and $\alpha_l$, if $\alpha_k > \alpha_l$, then we assume that $\phi_k < \phi_l$.    \QET
\end{Assumption}


To deal with all the attack types of $\Tc$ defined in \eqref{attack_type}, the defender evaluates all the possible disruptions on the network for every attack type. This evaluation enables the defender to optimally allocate sensors such that the expected disruption is minimized. As a result, given a set of attack types $\Tc$, the defender
solves the following optimization problem:
\begin{align}
    R(\Tc \,|\, \beta) &\triangleq 
      \min_{\Mc \, \subset \, \Vc, \, |\Mc| \leq \beta} ~ \sum_{\alpha_k \, \in \, \Tc} ~
     \phi_k  ~ U(\Mc\,|\,\alpha_k), \label{R_min}
\end{align}
where $U(\Mc\,|\,\alpha_k)$ is the defense worst-case cost corresponding to $\alpha_k$-type adversary who solves \eqref{Q_max}. This cost is defined as follows:
\begin{align}
    &U(\Mc\,|\,\alpha_k) \triangleq c_s(\Mc) 
    + Q(\Mc\,|\,\alpha_k), \label{Rk}
\end{align}
where the cost of utilized sensors $c_s(\Mc)$ is given in \eqref{cost_sensor} and $Q(\Mc\,|\,\alpha_k)$ is given in \eqref{Q_max}. In the following, we formulate the research problems that will be addressed in this paper.

\begin{Remark}
    [Game theory perspective]
    The adversary and defender models described in this section are adapted from game theory with incomplete information, which Harsanyi introduced in 1967 \cite{harsanyi1967games}. In this paper, the adversary finds the best response by solving \eqref{Q_max} while the defender, who is uncertain of the adversary's objective, minimizes the expected value of all the possible adversary's payoffs \eqref{R_min}. This information assumption is also considered in Bayesian games, applied in security problems \cite{umsonst2021bayesian,gupta2016dynamic}. It is worth noting that this game may not yield a pure equilibrium if the two players make decisions simultaneously. To deal with such an issue, the concept of Stackelberg games \cite{bacsar1998dynamic} is adopted in this paper, in line with \cite{li2018false,yuan2019stackelberg,shukla2022robust}, where the adversary is allowed to move after fully observing the defender's decision.    \QET
\end{Remark}

\subsection{Scalable and optimal security allocation problem}
In the remainder of this paper, we address the following two research problems addressing \eqref{R_min}:
\begin{Problem}
[Optimal Security Allocation]
    \label{prob:general_sol} Given the network control system under stealthy attacks \eqref{sys:xa} and the adversary and defender models, find an optimal solution to the defense cost \eqref{R_min}.
\end{Problem}
\begin{Problem}
[Scalable Algorithm]
    \label{prob:impatient_sol} Given that the networked control system \eqref{sys:xa} can be of very large dimension ($N \ggg 1$), develop a \emph{computationally scalable algorithm} to find an optimal solution to the defense cost \eqref{R_min}.
\end{Problem}

\section{Optimal Security Allocation}
\label{sec:general_solution}
In this section, we first analyze the worst-case disruption on the global performance \eqref{Q_sup} and then present its computation. As a result, the computation of \eqref{Q_sup} enables us to find the optimal attack policy that is the solution to \eqref{Q_max}.
In the remainder of the section, we provide an SDP problem that solves \textit{Problem~\ref{prob:general_sol}} to obtain the optimal set of monitor nodes for the defender. 
\subsection{Worst-case disruption computation}

The following lemmas state the existence and the computation of a solution to \eqref{Q_max} for every possible attack type $\alpha_k$.


\begin{Lemma}
    [Bounded disruption]
    \label{lem:Q_bounded}
    Let us consider the networked control system \eqref{sys:xa}-\eqref{sys:yma} with a monitor set $\Mc$ under a stealthy attack $\zeta(t)$ conducted by the $\alpha_k$-type adversary.
    Then, the worst-case impact \eqref{Q_max} always admits a finite solution.
    \QET
\end{Lemma}    
The proof of Lemma~\ref{lem:Q_bounded} follows from our previous result in \cite[Proposition 2.1]{anand2023risk}.
Next, we provide a method for solving \eqref{Q_max}. The following computation allows the malicious adversary to choose the optimal attack policy for a given monitor set decided by the defender.  
\begin{Lemma}
    [Worst-case disruption computation]
    \label{lem:Q_computation}
    Suppose that the monitor set $\Mc$ is given and the adversary is $\alpha_k$-type. For each admissible attack set $\Ac \in \Sc_{\alpha_k}$, 
    a tuple of variables $( \{\gamma_{\Ac}^m\}_{\forall m \, \in \, \Mc },\,\psi_\Ac,P_\Ac) \, \in \, \Rbb_{>0} \times \Rbb_{>0}^{\alpha_k} \times \Sbb^N_{\succeq 0}$ is defined correspondingly. The optimal attack policy of the $\alpha_k$-type adversary, which is formulated in \eqref{Q_max}, is computed by the following SDP problem:
    \begin{align}
    & Q(\Mc\,|\,\alpha_k) = \min_{Q_k,\, \{
    \{\gamma_{\Ac}^m\}_{\forall m \, \in \, \Mc},\,\psi_\Ac,\,P_\Ac \}_{\forall  \Ac \, \in \, \Sc_{\alpha_k}} }~~ Q_k \label{Q_sdp}  \\
    \text{s.t.}
    &~~~~~
    Q_k \in \Rbb_{>0}, ~
    \gamma_{\Ac}^m \in \Rbb_{>0},~ \psi_\Ac \in \Rbb_{>0}^{\alpha_k}, ~
    P_\Ac \in \Sbb^{N}_{\succeq 0}, \non \\
    &~~~~~
    \sum_{m \, \in \, \Mc}  \gamma_{\Ac}^m \, \delta_m  
    + E \, \textbf{1}^\top  \psi_\Ac  \leq Q_k, \non \\
    &~~~~~
    \ba{cc}
    -L^\top P_\Ac - P_\Ac L + W^2 & P_\Ac B_\Ac \\
    B_\Ac^\top P_\Ac & -\, \textbf{diag}(\psi_\Ac)
     \ea 
    \non \\
    &~~~~~
    - \sum_{m \, \in \, \Mc}   \textbf{diag} \Bigg(\ba{c} 
    \gamma_{\Ac}^m \, e_m \\ 0 \ea \Bigg) \preceq 0,
    \forall m \in \Mc, \forall \Ac \in \Sc_{\alpha_k}
    . \non 
    \end{align}
    Here, $\textbf{1}$ stands for an all-one vector with a proper dimension. \QET
\end{Lemma}
\begin{proof}
The optimization problem \eqref{Q_max} is equivalent to the following optimization problem:
    \begin{align}
        Q(\Mc\,|\,\alpha_k) = &\min_{Q_k \, > \, 0} ~~ Q_k \label{Qstar} \\
        \text{s.t.}&~~~
        V(\Mc, \Ac)  \leq Q_k, ~ \forall \, \Ac \in \Sc_{\alpha_k}. \non
    \end{align}
    On the other hand, for a given element $\Ac$ in $\Sc_{\alpha_k}$, we show how to compute the worst-case disruption \eqref{Q_sup}, which has the dual form:
    \begin{align}
        &\inf_{ 
        \{\gamma^m_\Ac\}_{\forall m \, \in \, \Mc} \,\in\, \Rbb_{>0} ,\, \psi_\Ac \,\in\, \Rbb^{\alpha_k}_{>0} } ~ \bigg[ ~ \sup_{ \zeta }~ \bigg\{
        \norm{p^a}_{\Lc_2}^2 
        \non \\
        &
        + \sum_{m \, \in \, \Mc} \gamma^m_\Ac \big( \delta_m - \norm{y_m^a}_{\Lc_2}^2 \big) 
        + \sum_{j\,=\,1}^{\alpha_k} 
    e_j^\top
    \psi_\Ac \big( E - \norm{e_j^\top \zeta}_{\Lc_2}^2 \big) \bigg\}\bigg] \label{Q_inf} \\
        & 
        \text{s.t.}~~ \eqref{sys:xa}-\eqref{sys:yma}, \, x^a(0) = 0,\, x^a(\infty) = 0, \non 
    \end{align}
    where $\gamma^m_\Ac$ and $\psi_\Ac$ are Lagrange multipliers associated with the first and second inequality constraints in \eqref{Q_sup}, respectively.
    The dual form \eqref{Q_inf} is bounded only if 
    \begin{align}
        &\norm{p^a}_{\Lc_2}^2 - 
        \sum_{m  \in \Mc} \gamma^m_\Ac  \norm{y_m^a}_{\Lc_2}^2 
        - \norm{\textbf{diag}(\psi_\Ac)^{\frac{1}{2}} \zeta}^2_{\Lc_2}   \leq  0, \non 
    \end{align}
    which results in the following optimization problem:
    \begin{align}
    V(\Mc,\Ac) = &
    \inf_{\{\gamma^m_\Ac\}_{\forall m \, \in \, \Mc} ,\, \psi_\Ac}  ~~ \sum_{m \, \in \, \Mc}  \gamma^m_\Ac \delta_m + E \, \,\textbf{1}^\top  \psi_\Ac  \label{Q_min} \\ 
    \text{s.t.}~~~& 
    \eqref{sys:xa}-\eqref{sys:yma}, \, x^a(0) = 0, \, x^a(\infty) = 0, \non \\
    &\gamma^m_\Ac \, \in \, \Rbb_{>0},\, \psi_\Ac \, \in \, \Rbb_{>0}^{\alpha_k}, \non \\
    &  \hspace{-1cm}
    \norm{p^a}_{\Lc_2}^2 - \sum_{m \, \in \, \Mc} \gamma^m_\Ac \norm{y_m^a}^2_{\Lc_2} - 
    \norm{\textbf{diag}(\psi_\Ac)^{\frac{1}{2}} \zeta}^2_{\Lc_2} \leq 0.
    \non 
    \end{align}
    The strong duality can be proven by utilizing the lossless S-Procedure \cite[Ch. 4]{petersen2000robust}. Recalling the key results in the dissipative system theory for linear systems \cite{trentelman1991dissipation} with a non-negative storage function $S(x^a) \triangleq (x^a)^\top P_\Ac x^a$, where $P_\Ac  \in \Sbb^N_{\succeq 0} $, and a supply rate $s(\cdot,\cdot) \triangleq   \sum_{m \, \in \, \Mc} \gamma^m_\Ac \norm{y_m^a}^2_{\Lc_2} + \norm{\textbf{diag}(\psi_\Ac)^{\frac{1}{2}} \zeta}^2_{\Lc_2} - \norm{p^a}_{\Lc_2}^2$, we observe that the inequality constraint in \eqref{Q_min} is equivalent to the system being dissipative with respect to the supply rate $s(\cdot,\cdot)$. Hence, the inequality constraint in \eqref{Q_min} can be replaced with the equivalent dissipation inequality and 
    the optimization problem \eqref{Q_min} is translated into the following SDP problem 
    \begin{align}
    \hspace{-0.1cm}
        V(\Mc,\Ac) = &
    \min_{\{\gamma^m_\Ac\}_{\forall m \, \in \, \Mc} ,\, \psi_\Ac, \,P_\Ac}~  \sum_{m \, \in \, \Mc}  \gamma^m_\Ac  \delta_m + E \,\textbf{1}^\top  \psi_\Ac  \label{Q_min_sdp} \\ 
    \text{s.t.}~~~& \gamma^m_\Ac \, \in \, \Rbb_{>0},\, \psi_\Ac \, \in \, \Rbb_{>0}^{\alpha_k}, \, P_\Ac \in \Sbb^N_{\succeq 0}, \non \\
    &
    \ba{cc}
    -L^\top P_\Ac - P_\Ac L + W^2 & P_\Ac B_\Ac \\
    B_\Ac^\top P_\Ac & -\, \textbf{diag}(\psi_\Ac)
    \ea \non \\
    &
    - \sum_{m \, \in \, \Mc}   \textbf{diag} \Bigg(\ba{c} 
    \gamma^m_\Ac \, e_m \\ 0 \ea \Bigg) \preceq 0. \non
    \end{align}
    Substituting \eqref{Q_min_sdp} into \eqref{Qstar} yields \eqref{Q_sdp}, which concludes the proof.
\end{proof}

Levering the results of Lemmas~\ref{lem:Q_bounded} and \ref{lem:Q_computation}, the following subsection proposes a solution to \textit{Problem~\ref{prob:general_sol}}.





\subsection{Solution to \textit{Problem~\ref{prob:general_sol}} (general security allocation)}
Before solving \textit{Problem~\ref{prob:general_sol}}, we present a property of $\gamma_\Ac^{m \star}$, which is a part of the solution to \eqref{Q_sdp}, in the following lemma.
\begin{Lemma}
    \label{lem:Q_inf}
    Consider the SDP problem \eqref{Q_sdp} in Lemma~\ref{lem:Q_computation}. For any $m \in \Mc$ and $\Ac \in \Sc_{\alpha_k}$, the optimal solution $\gamma_\Ac^{m \star}$ to  \eqref{Q_sdp} has the following upper bound:
    \begin{align}
        \gamma_\Ac^{m \star} \leq \left[ \min_{i \, \in \, \Vc} \delta_i \right]^{-1}
        \max_{\Ac \in \Sc_{\alpha_k}}  V_{\infty}(\Ac),
        \label{maxgamma}
    \end{align}
    where 
    \begin{align}
    V_{\infty}(\Ac) \triangleq V(\emptyset,\Ac). \label{Q_sup_im_Hinf}
    \end{align}
\end{Lemma}
\begin{proof}
    It is worth noting that \eqref{Q_sup_im_Hinf} corresponds to \eqref{Q_sup} when no monitoring node is considered, i.e., the first inequality constraint in \eqref{Q_sup} is removed, resulting in $V_\infty(\Ac) \geq V(\Mc,\Ac) ~ \forall \Ac$. As a result, one has
    $\max_{\Ac \in \Sc_{\alpha_k}}  V_{\infty}(\Ac) \geq Q(\Mc\,|\,\alpha_k)$, which is defined in \eqref{Q_max} and computed by \eqref{Q_sdp}. 
    On the other hand, from \eqref{Q_sdp}, one obtains
    $\sum_{m \, \in \, \Mc} \gamma_\Ac^{m\star} \delta_m \leq Q(\Mc\,|\,\alpha_k)$,
    which leads to \eqref{maxgamma}.
\end{proof}
Prior to the solution to \textit{Problem~\ref{prob:general_sol}}, let us make use of the following notation. Denote $\hat V_\infty \triangleq [\min_{i \, \in \, \Vc} \delta_i ]^{-1}  \max_{\alpha_k \, \in \, \Tc}  \max_{\Ac \in \Sc_{\alpha_k}}  V_{\infty}(\Ac)$ where $V_{\infty}(\Ac)$ is defined in \eqref{Q_sup_im_Hinf}. For each attack type $\alpha_k \in \Tc$, define a positive scalar variable $Q_k \in \Rbb_{>0}$ and recall $\Sc_{\alpha_k}$ defined in \eqref{col_attack_set}. For each admissible attack action $\Ac \in \Sc_{\alpha_k}$, a tuple of variables $(\omega_\Ac,\,\psi_\Ac,P_\Ac) \, \in \, \Rbb^N_{\geq 0} \times \Rbb_{>0}^{\alpha_k} \times \Sbb^N_{\succeq 0}$ is defined correspondingly.
Recall the vector of all the alarm thresholds $\delta$ and the vector of all the sensor costs $\kappa$.
Denote an $N$-dimensional binary vector $z$ as a representation of the monitor set $\Mc$ where $m$-entry of $z$ being equal to $1$ indicates that $m$-th node belongs to $\Mc$.
Based on these definitions and the results presented in Lemmas~\ref{lem:Q_bounded}-\ref{lem:Q_inf}, we are now ready to present the following theorem that solves \textit{Problem~\ref{prob:general_sol}}.
\begin{Theorem}[Optimal security allocation]
    \label{th:general_sol}
    Consider the networked control system \eqref{sys:xa}-\eqref{sys:yma} 
    associated with 
    a monitor set $\Mc$ under a stealthy attack $\zeta(t)$, which is assumed to be conducted by one of 
    $n_a$ attack types in $\Tc$ defined in \eqref{attack_type}. 
    Recall the monitor set $\Mc$ denoted as a binary variable $z \in \{0,1\}^N$.
    Then, the optimal monitor set, which is the solution to \eqref{R_min}, is determined by $z^\star$, which is the optimal solution to the following mixed-integer SDP problem:
    \begin{align}
    &\min_{z, \, \big\{Q_k,\{\omega_{\Ac},\, \psi_{\Ac},\, P_{\Ac}\}_{\forall \Ac \,\in\, \Sc_{\alpha_k}} \big\}_{\forall \alpha_k \in \Tc} } ~~ \kappa^\top z + 
    \sum_{\alpha_k \, \in \, \Tc}  \phi_k  \, Q_k    \label{R_sdp} \\
    &\text{s.t.}~
    z \in \{0,1\}^N, \, Q_k \in \Rbb_{>0}, \, \omega_{\Ac} \in \Rbb_{\geq 0}^N, \, \psi_{\Ac} \in \Rbb^{|\Ac|}_{>0}, 
    P_{\Ac} \in \Sbb^{N}_{\succeq 0}, \non \\
    &~~~~~
    \textbf{1}^\top z \leq \beta, \, \omega_{\Ac} \, \leq \, \hat V_\infty z, \,
    \delta^\top \omega_{\Ac} + E~ \textbf{1}^\top  \psi_{\Ac}  \, \leq \, Q_k, 
    \non \\ 
    &~~~~~
    \ba{cc}
    -L^\top P_{\Ac} - P_{\Ac} L + W^2 & P_{\Ac} B_\Ac \\
    B_\Ac^\top P_{\Ac} & -\, \textbf{diag}(\psi_{\Ac})
    \ea \non \\
    &~~~~
    - \textbf{diag} \Bigg(\ba{c} 
    \omega_{\Ac} \\ 0 \ea \Bigg) \preceq 0, 
    ~\forall \, \Ac \in \Sc_{\alpha_k}, \, \forall \alpha_k \in \Tc, \non
    \end{align}
    where 
    $\textbf{1}$ stands for an all-one vector with a proper dimension. 
    \QET
\end{Theorem}

The proof of Theorem~\ref{th:general_sol} is given in Appendix~\ref{app:thproof_general_sol}.
We highlight several interesting aspects of the results presented in this subsection in the following remarks.

\begin{Remark}
    [Diagonal norm-bound]
    It is worth noting that \eqref{Q_sup_im_Hinf} differs from the classical $\Hc_\infty$ norm since it considers the element-wise bounded energy for input signals. This metric is also called the diagonal norm-bound in \cite[Ch. 6]{boyd1994linear}.     \QET 
\end{Remark}
\begin{Remark}
    [Practical implementation]
    The value of $V_\infty(\Ac)$ in \eqref{Q_sup_im_Hinf} can be computed by solving \eqref{Q_min_sdp} with no monitor node and a given $\Ac$. By definition,
    the value of $\hat V_\infty$ in Theorem~\ref{th:general_sol} can be computed by evaluating $V_\infty(\Ac)$ for all possible $\Ac$. Alternatively, the value of $\hat V_\infty$
    can be chosen as a very large positive number without explicitly computing \eqref{Q_sup_im_Hinf} for all possible $\Ac$ in practice. 
    This value is also referred to as a ``big M'' in \cite{milovsevic2023strategic}. \QET
\end{Remark}
\begin{Remark}
    [Computational complexity] To solve Problem~\ref{prob:general_sol}, the solution \eqref{R_sdp} proposed in Theorem~\ref{th:general_sol} requires us to consider all admissible attack scenarios, significantly increasing the number of optimal decision variables. As a consequence, the computational complexity increases as the size of the network increases. It is worth noting that the same issue is witnessed in recent studies \cite{li2018false,yuan2019stackelberg,shukla2022robust,umsonst2021bayesian}. In an attempt to deal with such a computational issue, we conduct a deeper analysis of the worst-case disruption \eqref{Q_sup} to find a control design procedure that assists us in solving the security allocation problem more efficiently in the following section.
    \QET
\end{Remark}

\section{Scalable Allocation Algorithm}
\label{sec:impatient_adv}


In this section, we conduct an analysis on the computation of the worst-case disruption defined in \eqref{Q_sup}. The analysis enables us to employ the result presented in Theorem~\ref{th:general_sol} in a more efficient way, albeit under a suitably designed control policy.

Recall the definition of the worst-case disruption \eqref{Q_sup}, we first introduce the following definition of its feasible set in the following.
\begin{Definition}[Feasible set]
Given a pair of $(\Mc, \Ac)$ and the networked control system \eqref{sys:xa}-\eqref{sys:yma},
the set of feasible values of the optimization problem \eqref{Q_sup} 
is defined as 
\begin{equation}\label{Q_Ceas_set}
    \Qc(\Mc, \Ac)\triangleq \left\{ q \;\Bigg|\;~
    \begin{aligned}
    & 
    \norm{e_j^\top \zeta}^2_{\Lc_2} \leq E, \, \forall \, j  \in \{1,2,\ldots, |\Ac| \}\\
    & \norm{y_m^a}^2_{\Lc_2} \leq \delta_m, \norm{p^a}_{\Lc_2}^2 \leq q\\
    &x^a(0) = 0, x^a(\infty) = 0
    \end{aligned}
    \right\}. 
\end{equation} \QET 
\end{Definition}

With the above definition of feasible sets, we organize the remainder of this section as follows. First, under a new assumption (say $A1$), we show that any feasible solution to the optimization problem \eqref{Q_sup} admits an upper bound (Lemma 5). We provide a control design procedure under which the assumption $A1$ is satisfied (Lemma 6). Then, we show that the solution to the optimization problem \eqref{Q_sup} is computationally scalable (Theorem 2) under assumption $A1$ and the results of Lemma 5. As a result, the optimal attack policy and the optimal security allocation problem can be inherently computationally scalable (Corollaries 2 and 3).

\begin{Lemma}\label{lem:WboundQ}
Let $\Qc(\Mc, \Ac)$ denote the feasible set of \eqref{Q_sup} and suppose that the following condition \eqref{cond:W2_Hinf} holds
    \begin{align}
        W^2 \succeq \left[ \min_{m \in \Mc} \delta_m \right]^{-1} V_{\infty}(\Ac) ~I,
        \label{cond:W2_Hinf}
    \end{align} 
where $V_{\infty}(\Ac)$ is defined in \eqref{Q_sup_im_Hinf} and $I$ is an identity matrix.
Then, one obtains $W^2 \succeq  \left[ \min_{m \in \Mc} \delta_m \right]^{-1} q \, I,~\forall \, q \in \Qc(\Mc, \Ac)$. \QET
\end{Lemma}
\begin{proof}
    By construction, the optimal solution to \eqref{Q_sup_im_Hinf} is not smaller than that to \eqref{Q_sup}, i.e., $V_{\infty}(\Ac) \geq V(\Mc,\Ac)$. 
    If the condition \eqref{cond:W2_Hinf} holds, then $W^2 \succeq \left[ \min_{m \in \Mc} \delta_m \right]^{-1} V(\Mc,\Ac) \, I$, concluding the proof.
\end{proof}
\begin{Lemma}
\label{lem:self_loop}
    Consider the networked control system \eqref{sys:xa} associated with a graph $\Gc = (\Vc,\Ec,A,\Theta)$ where positive diagonal matrix $\Theta$ stands for the self-loop control gains at all the nodes.
    There exists a suitable diagonal matrix $\Theta$ such that the condition \eqref{cond:W2_Hinf} holds.  \QET
\end{Lemma}

The proof of Lemma~\ref{lem:self_loop} is given in Appendix~\ref{app:lempf_self_loop}. 
Based on the results of Lemmas~\ref{lem:WboundQ} and \ref{lem:self_loop} and the result in \cite[Thm. 1]{rantzer2015kalman}, the worst-case disruption \eqref{Q_sup} can be computed in a more efficient way, which is presented in the following theorem.

\begin{Theorem}
[Scalable assessment]
\label{th:scalable:Q}
Let $\Qc(\Mc, \Ac)$ denote the set of feasible values of \eqref{Q_sup} and suppose that the condition \eqref{cond:W2_Hinf} holds. Then, the optimal value of the optimization problem \eqref{Q_sup} is obtained by solving the convex SDP \eqref{Q_min_sdp} with a positive definite diagonal matrix variable $P_{\Ac}$. \QET
\end{Theorem}


\begin{proof} 
See Appendix~\ref{app:th_anders}.
\end{proof}
Theorem \ref{th:scalable:Q} states that the worst-case disruption \eqref{Q_sup} can be obtained by solving the SDP \eqref{Q_min_sdp} with a diagonal matrix $P_\Ac$. Thus, we can infer that, under an enforceable assumption stated in \eqref{cond:W2_Hinf}, the computational complexity of \eqref{Q_sup} increases linearly with the size of the network, for any given $\Ac$. 

To explain this briefly, consider a network with $N$ nodes.
Then, in general, the optimization problem \eqref{Q_min_sdp} has $\frac{N(N+1)}{2}+|\mathcal{M}| + |\mathcal{A}|$ variables. To detail, matrix $P_\Ac$ has $\frac{N(N+1)}{2}$ variables, and the Lagrange multipliers $\gamma^m_\Ac$ and $\psi_\Ac$ account for the remaining variables, respectively. On the other hand, our proposed solution in \textit{Theorem \ref{th:scalable:Q}} has only $N+|\mathcal{M}| + |\mathcal{A}|$ variables. This reduction in the number of variables stems from the fact that the matrix ${P}_\Ac$ is diagonal. 

This complexity reduction allows for our optimal allocation algorithm proposed in \textit{Theorem~\ref{th:general_sol}} to be scalable. More specifically, the following corollaries are presented as a solution to \textit{Problem~\ref{prob:impatient_sol}}.

\begin{Corollary}
[Scalable optimal attack policy]
    \label{Cor:Op_Attack_scalable}
    Suppose the condition \eqref{cond:W2_Hinf} holds for all possible attack sets $\Ac$. The solution to \eqref{Q_max} is given by solving \eqref{Q_sdp} where $P_{\Ac}$ is a diagonal matrix variable for all attack sets $\Ac$. \QET
\end{Corollary}

\begin{Corollary}
[Scalable security allocation]
    \label{Cor:Se_Allo_scalable}
Suppose the condition \eqref{cond:W2_Hinf} holds for all possible attack sets $\Ac$. Then, the solution to the security allocation problem \eqref{R_min} is given by solving the convex SDP \eqref{R_sdp} 
where $P_{\Ac}$ is a diagonal matrix.\QET
\end{Corollary}


\begin{Remark}
    The computation to guarantee the satisfaction of \eqref{cond:W2_Hinf} in Colloraries~\ref{Cor:Op_Attack_scalable}-\ref{Cor:Se_Allo_scalable} can be performed by parallel computations.     \QET
\end{Remark}


\section{Numerical Examples}
\label{sec:simulation}
In this section, we validate the obtained main results which are presented in \textit{Theorem~\ref{th:general_sol}}, \textit{Theorem~\ref{th:scalable:Q}}, and \textit{Corollary~\ref{Cor:Se_Allo_scalable}} with different network sizes. 
The simulation is performed using Matlab 2023b with YALMIP 2023 toolbox \cite{lofberg2004yalmip} and MOSEK solver.


\subsection{Optimal security allocation: solution to \textit{Problem~\ref{prob:general_sol}}}
We validate the obtained result of \textit{Theorem~\ref{th:general_sol}} by comparing the solution obtained by solving \eqref{R_sdp} with the solution from an exhaustive-search method. The exhaustive-search method computes the defense cost \eqref{R_min} for all the admissible monitor sets and gets the one that yields the minimum defense cost. The comparison result is depicted in Figure~\ref{fig:10Nodeopt} where the defense cost computed by \eqref{R_sdp} is represented by a red star while the ones computed by the exhaustive-search method are illustrated by black dots in each experiment.

In Figure~\ref{fig:10Nodeopt},
we consider 20 Erdős–Rényi random directed connected graphs where an edge is included to connect two vertices with a probability of 0.25, in which each network contains 10 nodes ($N = 10$). The parameters of the networks can be chosen as follows: $\theta_i = 0.7$, $\delta_i = 0.5$, $w_i = 1$, and $\kappa_i = 0.3$ for all $i \in \Vc$; the sensor budget $\beta = 3$; the maximum attack energy $E = 10$; the number of attack types $n_a = 3$;
the set of attack types $\Tc = \{ \alpha_1, \, \alpha_2, \, \alpha_3 \}$ where $\phi_1 = 0.5$, $\phi_2 = 0.35$, and $\phi_3 = 0.15$ (\textit{Assumption~\ref{assumption:adversary_type}} holds).

By observing all the experimental results in Figure~\ref{fig:10Nodeopt}, the red stars (solution to \eqref{R_sdp}) consistently align with the minimum values indicated by the black dots (solution given by the exhaustive-search method). Numerically, the relative deviation between the optimal solution by solving \eqref{R_sdp} and the minimum values obtained from the exhaustive-search method lies in the range $[-3.13, \, 1.32]\times 10^{-4}\%$.
This result shows that the SDP problem \eqref{R_sdp} enables us to successfully find the optimal set of monitor nodes, thereby corroborating the validity of \textit{Theorem~\ref{th:general_sol}}. 

\begin{figure}[!t]
    \centering
    \includegraphics[width=0.45\textwidth]{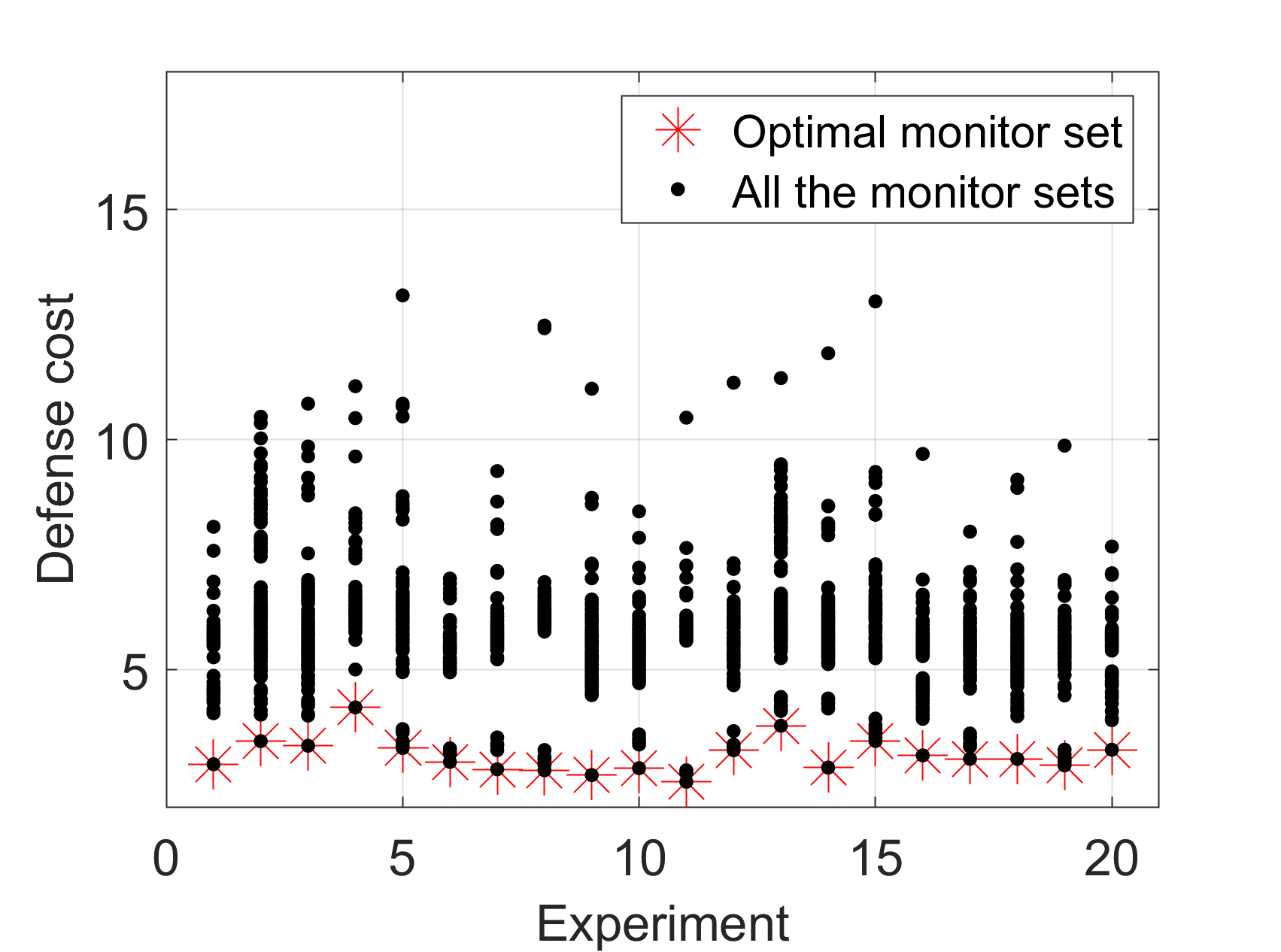}
    \caption{The red stars are the optimal values found by Theorem~\ref{th:general_sol} while black dots are found by an exhaustive-search method. The figure shows that the optimal defense costs found by Theorem~\ref{th:general_sol} are extremely close to the minimum defense costs found by the exhaustive-search method. In more detail, the relative deviation between them lies in the range $[-3.13, \, 1.32]\times 10^{-4}\%$.}
    \label{fig:10Nodeopt}
\end{figure}
\subsection{Scalable algorithm: solution to \textit{Problem~\ref{prob:impatient_sol}}}
In this part, we examine the results obtained from Theorem~\ref{th:scalable:Q} and Corollary~\ref{Cor:Se_Allo_scalable} to see how the obtained results reduce the computation complexity. The computation complexity is assumed to correlate with the computation time used to solve optimization problems. We investigate the computation time in solving \eqref{Q_min_sdp} and \eqref{R_sdp} with different network sizes. 
The numerical results are illustrated in Figure~\ref{fig:comparison}. 
To validate the results of \textit{Theorem~\ref{th:scalable:Q}} and \textit{Corollary~\ref{Cor:Se_Allo_scalable}}, we solve \eqref{Q_min_sdp} and \eqref{R_sdp} under two scenarios. The first scenario utilizes $P_\Ac$ as symmetric positive semi-definite matrices while the second scenario uses $P_\Ac$ as diagonal positive matrices.
It is worth noting that the hardware limitation only allows us to solve \eqref{R_sdp} with a network of at most 70 nodes in the first scenario. Meanwhile, solving \eqref{R_sdp} with a network of more than 70 nodes in the second scenario is possible. However, results are not reported due to a lack of comparison.
According to the results of \textit{Theorem~\ref{th:scalable:Q}}, the two scenarios obtain the same optimal value even though the solved matrices $P_\Ac$ are different. As discussed in the paragraph below \textit{Theorem~\ref{th:scalable:Q}}, the second scenario involves significantly fewer variables in the optimization problem, thus significantly reducing the computation complexity. Therefore, the computation time consumed by the second scenario should be less than that consumed by the first scenario. Indeed, we show the comparison of the computation time consumed by the two scenarios with different network sizes in Figure~\ref{fig:comparison}. These figures demonstrate that using the second scenario (diagonal matrix variables $P_\Ac$) gives us the same optimal value while consuming significantly less computation time. More specifically, for $N = 500$ in Figure~\ref{fig:analysisdisruption}, solving \eqref{Q_min_sdp} took 91.39 minutes and 0.08 minutes on average in the first and second scenarios, respectively. For $N=70$ in Figure~\ref{fig:securityallocationcompare}, solving \eqref{R_sdp} took 167.38 minutes and 2.85 minutes on average in the first and second scenarios, respectively.

\begin{figure}[!t]
    \centering
    \begin{subfigure}
        {0.45\textwidth}
        \includegraphics[width=\textwidth]{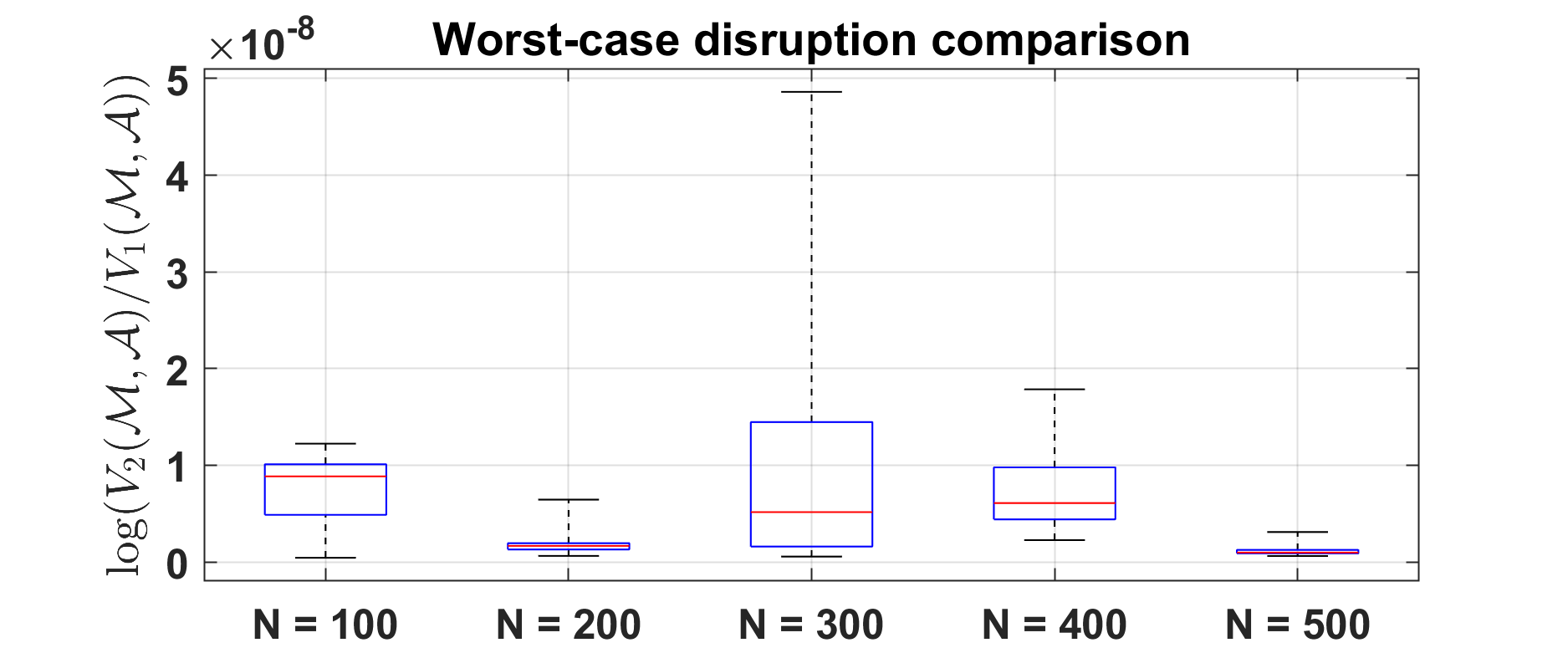}
    \end{subfigure}
    \begin{subfigure}
        {0.45\textwidth}
        \includegraphics[width=\textwidth]{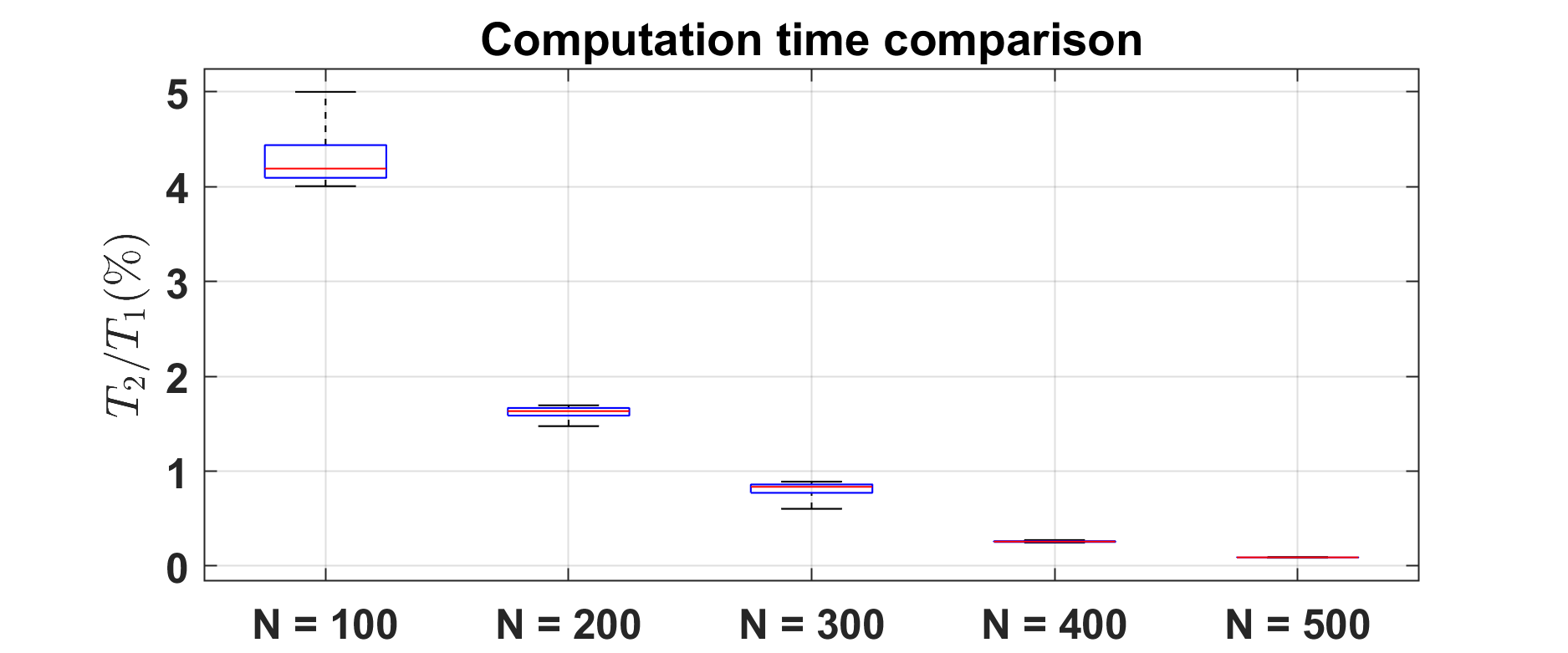}
        \caption{} \label{fig:analysisdisruption}
    \end{subfigure}
    \begin{subfigure}
        {0.45\textwidth}
        \includegraphics[width=\textwidth]{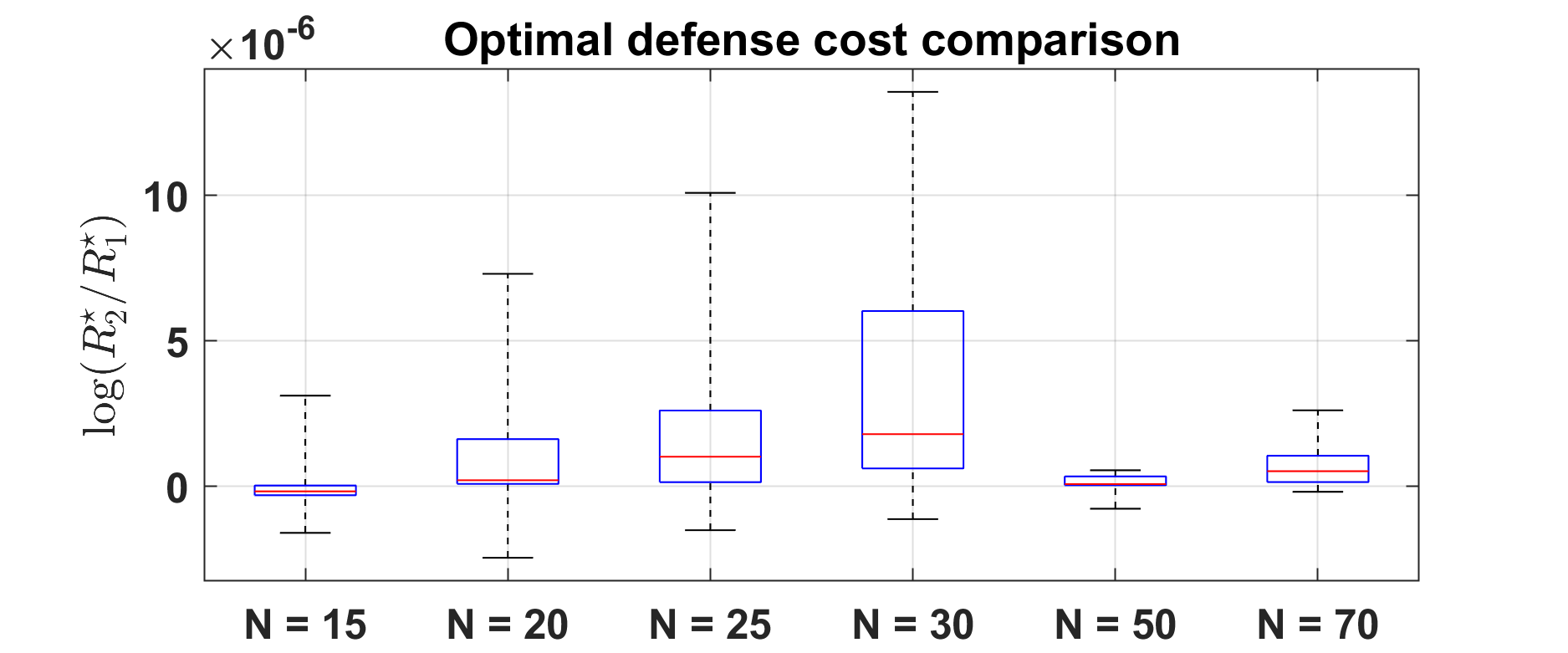}
    \end{subfigure}
    \begin{subfigure}
        {0.45\textwidth}
        \includegraphics[width=\textwidth]{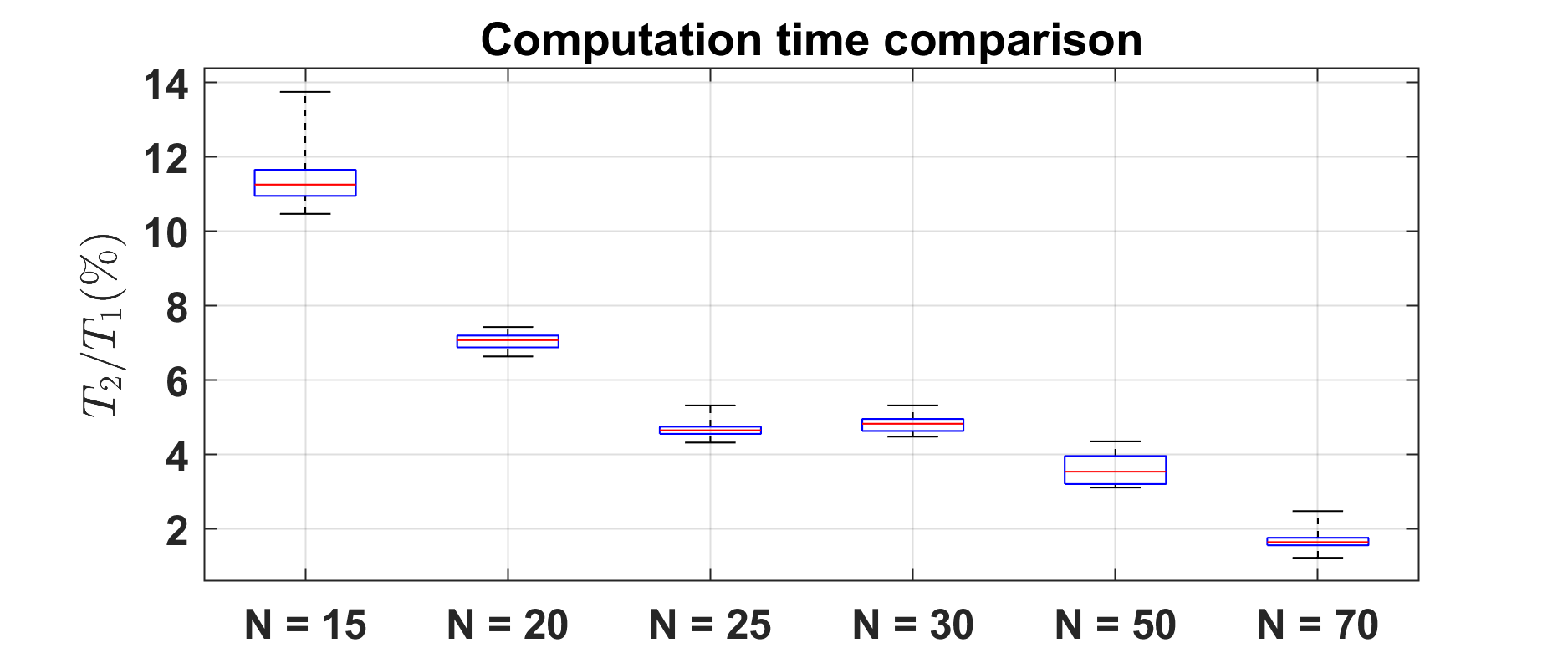}
        \caption{} \label{fig:securityallocationcompare}
    \end{subfigure}
    \caption{Comparison in finding the worst-case disruption \eqref{Q_min_sdp} (a) and in finding the optimal defense cost \eqref{R_sdp} (b) by the two scenarios where the first scenario solves \eqref{Q_min_sdp} and \eqref{R_sdp} with symmetric $P_\Ac$  while the second scenario solves the same optimization problems with diagonal $P_\Ac$. The notations $\{ V_1(\Mc,\Ac), T_1, R^\star_{1} \}$ represent the result obtained from the first scenario while $\{ V_2(\Mc,\Ac), T_2, R^\star_{2} \}$ represent the result obtained from the second scenario.
     The figures show that the optimal values obtained from the two scenarios are almost identical while the computation time consumed by the second scenario is extremely lower than that of the first scenario.}
    \label{fig:comparison}
\end{figure}

\section{Conclusion}
\label{sec:conclusion}
We investigated the security allocation problem in a networked control system under stealthy attacks. The uncertainty of the adversary's resources allowed us to formulate the objective of the defender as an expected defense cost based on the defender's rational information. Then, we proposed a method that supports the defender in optimally allocating the defense resources based on the expected defense cost. Moreover, by deeply investigating the worst-case disruption, an algorithm is presented 
to analyze the impact of stealthy attacks and to find the optimal monitor set in a scalable manner. 

~
\bibliographystyle{IEEEtran}
\bibliography{mybibfile}

\renewcommand{\thesection}{A.\arabic{section}}
\setcounter{section}{0}  
%
%
\section*{Appendix}
\section{Proof of Theorem \ref{th:general_sol}}
\label{app:thproof_general_sol}
From \eqref{R_min} and \eqref{Rk}, one has
\begin{align}
    &R(\Tc \,|\,\beta) = \min_{\Mc,\,|\Mc| \, \leq \, \beta}~ \sum_{\alpha_k \, \in \, \Tc} \phi_k  \, \bigg[\, c_s(\Mc) + Q(\Mc\,|\,\alpha_k) \bigg]. \label{R_min_alternative}
\end{align}
Next, we adopt the computation of $Q(\Mc\,|\,\alpha_k)$ in \textit{Lemma~\ref{lem:Q_computation}}. 
The SDP problem \eqref{Q_sdp} enables us to find the worst-case disruption \eqref{Q_max} for a given monitor set $\Mc$. In an attempt to find the best monitor set by substituting \eqref{Q_sdp} into \eqref{R_min_alternative}, we get into trouble with the nonlinearity of $\gamma^m_{\Ac} \, e_m$ in \eqref{Q_sdp} where $e_m$ can be represented as a variable node. This issue can be resolved by the following change of variables.

Given that the $N$-dimensional binary vector $z \in \{0,\,1\}^N$ is a representation of the monitor set $\Mc$, it has the following relationship:
\begin{align}
    z = \sum_{m \, \in \, \Mc} ~ e_m. \label{zM}
\end{align}
On the other hand, the sensor budget \eqref{sensor_budget} and the cost of utilized sensors \eqref{cost_sensor} imply the following constraints
\begin{align}
    |\Mc| = \textbf{1}^\top z \leq \beta, ~
    c_s(\Mc) = \sum_{m \, \in \, \Mc} \kappa^\top e_m = \kappa^\top z. \label{zM_cost}
\end{align}
Next, for a given attack action ${\Ac}$, an $N$-dimensional non-negative vector $\omega_{\Ac} \in \Rbb_{\geq 0}^N$ represents the following:
\begin{align}
    \omega_{\Ac} = \sum_{m \, \in \, \Mc}~\gamma^m_{\Ac} \, e_m, \label{oM}
\end{align}
which results in the relationship between the variables $\gamma^m_\Ac$ and $\omega_{\Ac}$:
\begin{align}
    \sum_{m \, \in \, \Mc}~\gamma^m_{\Ac} \, \delta_m  =  \sum_{m \, \in \, \Mc}~\gamma^m_{\Ac} \, (\delta^\top e_m) = \delta^\top \omega_{\Ac}. \label{oMdelta}
\end{align}
Intuitively, 
by sharing the same monitor set $\Mc$ and leveraging the result in \textit{Lemma~\ref{lem:Q_inf}}, the variables $z$ and $\omega_{\Ac}$ satisfy:
\begin{align}
    \omega_{\Ac} \, \leq \, \hat V_\infty z, \, \forall \, {\Ac} \in \Sc_{\alpha_k}, \, \forall \, \alpha_k \in \Tc. \label{bigM}
\end{align}
Finally, by substituting \eqref{Q_sdp}, 
\eqref{zM_cost}-\eqref{bigM} into \eqref{R_min_alternative}, one obtains the mixed-integer SDP problem \eqref{R_sdp}. Moreover, since the optimization problem \eqref{Q_sdp} always admits a finite solution (see \textit{Lemma~\ref{lem:Q_bounded}}) and $\kappa^\top z \leq \beta \max(\kappa) < \infty$, the mixed-integer SDP problem \eqref{R_sdp} always admits a finite solution. \QEDB
\section{Proof of Theorem \ref{th:scalable:Q}}
\label{app:th_anders}
We next provide a result that aids the proof of Theorem \ref{th:scalable:Q}.
\begin{Proposition}\label{thm_anders}
Let us consider a continuous LTI system $\Sigma$ with a state-space model: $\dot x(t) = A x(t) + B u(t),~y_1(t) = C_1 x(t),~y_2(t) = C_2 x(t)$ 
where $A$ is a Metzler matrix,
${B} \in \Rbb^{n \times m}_{\geq0}, {C_1} \in \Rbb^{p_1 \times n}_{\geq0},$ and ${C_2} \in \Rbb^{p_2 \times m}_{\geq0}$
Let us consider a quadratic supply rate {$s(\cdot,\cdot) = ||\Gamma^{\frac{1}{2}} u||^2 + \norm{y_1}^2 - \norm{y_2}^2$ where $\Gamma$ is a positive diagonal matrix} and let $C_2^\top C_2 \geq C_1^\top C_1$ where the inequality applies element-wise. Then the following statements are equivalent. 
\begin{enumerate}
    \item For all trajectories of the system with $x(0)=0,$ we have $\int_{0}^T s(\cdot,\cdot) \geq 0$;
    \item There exists a diagonal matrix $P \in \Sbb^n_{>0}$ such that 
    \begin{equation}\label{LMI_anders}
        \begin{bmatrix}
            A^\top P + PA & PB \\ B^\top P & 0 
        \end{bmatrix} + M \leq 0,
    \end{equation}
    where $M = \begin{bmatrix}
        C_2^\top C_2 - C_1^\top C_1 & 0\\ 0 & {-\Gamma}
    \end{bmatrix}$.\QET
\end{enumerate}
\end{Proposition}
\begin{Remark}
    Proposition~\ref{thm_anders} follows from \cite[Theorem 1]{rantzer2015kalman} which is a specialized result of the KYP lemma for positive systems. However, the time-domain interpretation of the KYP lemma was not explicitly stated in the theorem statement. Readers interested in the time domain interpretation are referred to \cite{megretski2000lecture}. \QET
\end{Remark}

\textit{Proof of Theorem~\ref{th:scalable:Q}:} 
The proof will recall a result from \cite{rantzer2015kalman} (see Appendix~\ref{app:th_anders} for more details).
Consider the optimization problem \eqref{Q_sup} which is equivalent to the following dual problem 
    \begin{align}\label{eq:dual:pf}
    & \min_{ \{\gamma^m_{\Ac}\}_{\forall m \in \Mc},\, \psi_{\Ac}} ~~ \sum_{m \, \in \, \Mc} ~ \gamma^m_{\Ac}  \delta_m + E \, \textbf{1}^\top  \psi_{\Ac} \\ 
    \text{s.t.}~& 
    \eqref{sys:xa}-\eqref{sys:yma}, \, x^a(0) = 0, x^a(\infty) = 0, \gamma^m_{\Ac} \in \Rbb_{>0}, \psi_{\Ac} \in \Rbb^{|\Ac|}_{>0}, \non \\
    &  
    \norm{p^a}_{\Lc_2}^2 - \sum_{m \, \in \, \Mc} \gamma^m_{\Ac} \norm{y_m^a}^2_{\Lc_2} 
    - 
    \norm{\textbf{diag}(\psi_\Ac)^{\frac{1}{2}} \zeta}^2_{\Lc_2} \leq 0.
    \non
\end{align}
Let us define the following supply rate 
\begin{align}\label{eq:supply}
    \!\!\!
    s^a(\cdot,\cdot) \triangleq  \norm{\textbf{diag}(\psi_\Ac)^{\frac{1}{2}} \zeta}^2_{\Lc_2} \!+\! \sum_{m  \in  \Mc} \gamma^m_{\Ac} \norm{y_m^a}^2_{\Lc_2}  - \norm{ p^a}^2_{\Lc_2},
\end{align}
where $p^a(t)$ and $y_m^a(t)$ are the outputs of the system $\Sigma^a \triangleq (-L,B_\Ac,\begin{bmatrix}
    W\\\sum_{m\in \mathcal{M}} \sqrt{\gamma^m_{\Ac}} \, \textbf{diag}(e_m)
\end{bmatrix},0)$ and $W = \textbf{diag}([w_i])$. Since the system $\Sigma^a$ is internally positive, if $\Sigma^a$ satisfies $W^2 \succeq \sum_{m\in \mathcal{M}} \gamma^m_{\Ac} \, \textbf{diag}(e_m)$, the results of Proposition \ref{thm_anders} readily applies to our setting. Next we show that any feasible solution $\gamma^m_{\Ac}$ satisfies $W^2 \succeq \sum_{m\in \mathcal{M}} \gamma^m_\Ac \, \textbf{diag}(e_m)$: let this be hypothesis 1 (H1).

We know that 
$V_{\infty}(\Ac) \geq q,\, \forall q \in \Qc(\Mc, \Ac)$ defined in \eqref{Q_Ceas_set}. However, since the optimization problems \eqref{Q_sup} and \eqref{eq:dual:pf} are equivalent, it follows that 
\begin{align}
    V_{\infty}(\Ac) &\geq  \sum_{m \in \Mc} \gamma^m_\Ac \delta_m + E \, \textbf{1}^\top \psi_\Ac > \sum_{m \in \Mc} \gamma^m_\Ac \delta_m \non \\
    &> \min_{m \in \Mc} \delta_m \sum_{m \in \Mc} \gamma^m_\Ac > \min_{m \in \Mc} \delta_m ~ \max_{m \in \Mc} \gamma^m_\Ac, \non \\
\text{resulting in} \;&\left[ \min_{m \in \Mc} \delta_m \right]^{-1}~V_{\infty}(\Ac)  > \max_{m \in \Mc} \gamma^{m\star}_\Ac \label{eq2}
\end{align}

Combining \eqref{cond:W2_Hinf} and \eqref{eq2}, it follows that $W^2 \succeq \left[ \min_{m \in \Mc} \delta_m \right]^{-1} V_{\infty}(\Ac) \, I \succ \max_{m \in \Mc} \gamma^{m\star}_\Ac \, I$: proving H1. Thus, the results of Proposition \ref{thm_anders} apply to our problem setting.

Consider the supply rate defined in \eqref{eq:supply}. Then, the constraint of the optimization problem \eqref{eq:dual:pf} can be rewritten as $\int s^a(\cdot,\cdot) \geq 0.$
Using the equivalence between $1)$ and $2)$ of Proposition \ref{thm_anders}, we can rewrite the constraint $\int s^a(\cdot,\cdot) \geq 0$ as \eqref{LMI_anders} using the system matrices of $\Sigma_a$. This concludes the proof.  
\QEDB

\section{Proof of Lemma~\ref{lem:self_loop}}
\label{app:lempf_self_loop}
Let us assume that \eqref{cond:W2_Hinf} does not hold. In the following, we show that an increase of all the local self-loop gains in $\Theta$ strictly reduces the value of $V_{\infty}(\Ac)$, resulting in the fulfillment of \eqref{cond:W2_Hinf}. 

Let us add a positive diagonal matrix $\Delta \Theta \succ 0$ to the self-loop gain matrix $\Theta$, resulting in the changed Laplacian matrix, which can be denoted as $\tilde L \triangleq L + \Delta \Theta$. The problem \eqref{Q_sup_im_Hinf} corresponding to the changed Laplacian matrix $\tilde L$ can be computed by the following SDP, which is adapted from \eqref{Q_min_sdp}:
    \begin{align}
    &\tilde V_{\infty}(\Ac) = \min_{\psi_{\Ac} \, \in \, \Rbb^{|\Ac|}_{>0}, \, P^\Ac \, \in \, \Sbb^N_{\succeq 0}} ~  E ~ \textbf{1}^\top  \psi_{\Ac} \label{Q_sup_im_Hinf_SDP_hat} \\
    \text{s.t.}&~ \ba{cc}
    -\tilde L^\top P_{\Ac} - P_{\Ac} \tilde L + W^2 & P_{\Ac} B_\Ac \\
    B_\Ac^\top P_{\Ac} & -\, \textbf{diag}(\psi_{\Ac})
    \ea   \preceq 0. \non 
\end{align}
    Since $\tilde L$ is Hurwitz and $W$ is a positive diagonal matrix, the necessary condition 
    for \eqref{Q_sup_im_Hinf_SDP_hat} is $P_\Ac \succ 0$ (top left block).
    Thus, the constraint in \eqref{Q_sup_im_Hinf_SDP_hat} can be rewritten as follows:
    \begin{align*}
        \ba{cc}
    -L^\top P_{\Ac} - P_{\Ac}  L + \tilde W^\top \tilde W & P_{\Ac} B_\Ac \\
    B_\Ac^\top P_{\Ac} & -\, \textbf{diag}(\psi_{\Ac})
    \ea   \preceq 0,
    \end{align*}
    where $\tilde W^\top \tilde W = W^2 - \Delta \Theta  P^\Ac - P^\Ac \Delta \Theta \prec W^2$ by the symmetric property. 
    Then, the optimization problem \eqref{Q_sup_im_Hinf_SDP_hat} is the dual form of the following optimization problem, which is adapted from \eqref{Q_sup}:
    \begin{align}
    \tilde V_{\infty}( \Ac) \triangleq & \sup_{ \zeta}~
     \norm{\tilde W  W^{-1}  p^a}_{\Lc_2}^2 \label{Q_sup_im_Hinf_hat}
     \\
    \text{s.t.}~~~&
    \eqref{sys:xa}-\eqref{sys:yma}, \, x^a(0) = 0, \, x^a(\infty) = 0,\non \\
    &\norm{e_j^\top \zeta}^2_{\Lc_2} \leq E,~ \forall \, j  \in \{1,2,\ldots,|\Ac| \}. \non 
\end{align}
    From $\tilde W^\top \tilde W \prec W^2$, \eqref{Q_sup_im_Hinf}, and \eqref{Q_sup_im_Hinf_hat}, one has $\tilde V_{\infty}(\Ac) < V_{\infty}(\Ac)$ by construction, which completes the proof. \QEDB
}

\end{document}

%% file: deflatex3.tex
\def\Ac{{\mathcal A}}

\def\Ec{{\mathcal E}}

\def\Gc{{\mathcal G}}

\def\Hc{{\mathcal H}}

\def\Lc{{\mathcal L}}

\def\Mc{{\mathcal M}}

\def\Nc{{\mathcal N}}

\def\Qc{{\mathcal Q}}

\def\Rbb{{\mathbb R}}

\def\Sc{{\mathcal S}}

\def\Sbb{{\mathbb S}}

\def\Tc{{\mathcal T}}

\def\Vc{{\mathcal V}}

\def\0{{\bf 0}}

\newcommand{\bitem}{\begin{itemize}}
\newcommand{\eitem}{\end{itemize}}
\newcommand{\btabular}{\begin{tabular}}
\newcommand{\etabular}{\end{tabular}}
\newcommand{\bcenter}{\begin{center}}
\newcommand{\ecenter}{\end{center}}
\newcommand{\bea}{\begin{eqnarray}}
\newcommand{\eea}{\end{eqnarray}}
\newcommand{\bean}{\begin{eqnarray*}}
\newcommand{\eean}{\end{eqnarray*}}
\newcommand{\ba}{\left[ \begin{array}}
\newcommand{\ea}{\\ \end{array} \right]}
\newcommand{\bear}{\begin{array}}
\newcommand{\eear}{\\ \end{array}}

\newcommand{\non}{\nonumber}

\newcommand{\ra}{\rightarrow}

\newcommand*{\QEDB}{\hfill\ensuremath{\blacksquare}}%
\newcommand*{\QET}{\hfill\ensuremath{\triangleleft}}
\newcommand{\norm}[1]{\left\lVert#1\right\rVert}

\newcounter{subequation}
\def\beasub{\addtocounter{equation}{+1}
\setcounter{subequation}{\value{equation}}
\setcounter{equation}{0}
\renewcommand{\theequation}{\arabic{subequation}\alph{equation}}
\begin{eqnarray}}
\def\eeasub{\end{eqnarray}
\setcounter{equation}{\value{subequation}}
\renewcommand{\theequation}{\arabic{equation}}}

\def\inf{\operatornamewithlimits{inf\vphantom{p}}}
